\documentclass[11pt]{amsart}

\usepackage{amssymb}
\usepackage{url}
\usepackage{mathrsfs}

	\normalbaselines						  %

\newcommand{\R}{{\mathbb  R}}

\newcommand{\D}{{\mathbb  D}}

\newcommand{\T}{\mathbb{T}}

\newcommand{\te}{{\theta}}

\newcommand{\Z}{{\mathbb  Z}}

\newcommand{\N}{{\mathbb  N}}

\newcommand{\C}{{\mathbb  C}}

\newcommand{\cK}{{\mathscr K}}
\newcommand{\rK}{{\mathcal K}}

\newcommand{\ID}{{\mathbf{1}}}

\newcommand{\OID}{{\mathbf{I}}}

\newcommand{\fdot}{\,\cdot\,}

\makeatletter
\def\Ddots{\mathinner{\mkern1mu\raise\p@
\vbox{\kern7\p@\hbox{.}}\mkern2mu
\raise4\p@\hbox{.}\mkern2mu\raise7\p@\hbox{.}\mkern1mu}}
\makeatother

\newcommand{\cH}{\mathcal{H}}

\newcommand{\f}{\varphi}

\newcommand{\p}{\mathbb{P}}
\newcommand{\oa}{\Omega_\alpha}
\newcommand{\sa}{\Sigma_\alpha}

\DeclareMathOperator{\PW}{PW}
\DeclareMathOperator{\re}{Re}

\DeclareMathOperator{\spa}{span}

\DeclareMathOperator{\clos}{clos}
\DeclareMathOperator{\Mod}{mod}
\DeclareMathOperator{\supp}{supp}

\newcommand{\ci}[1]{_{ {}_{\scriptstyle #1}}}
\newcommand{\ti}[1]{_{\scriptstyle \text{\rm #1}}}

\catcode`@=11
\chardef\mathlig@atcode\count255

\def\actively#1#2{\begingroup\uccode`\~=`#2\relax\uppercase{\endgroup#1~}}
\def\mathlig@gobble{\afterassignment\mathlig@next@cmd\let\mathlig@next= }
\def\mathlig@delim{\mathlig@delim}
\def\mathlig@defcs#1{\expandafter\def\csname#1\endcsname}
\def\mathlig@let@cs#1#2{\expandafter\let\expandafter#1\csname#2\endcsname}
\def\mathlig@appendcs#1#2{\expandafter\edef\csname#1\endcsname{\csname#1\endcsname#2}}

\def\mathlig#1#2{\mathlig@checklig#1\mathlig@end\mathlig@defcs{mathlig@back@#1}{#2}\ignorespaces}

\def\mathlig@checklig#1#2\mathlig@end{%
 \expandafter\ifx\csname mathlig@forw@#1\endcsname\relax
 \expandafter\mathchardef\csname mathlig@back@#1\endcsname=\mathcode`#1%
 \mathcode`#1"8000\actively\def#1{\csname mathlig@look@#1\endcsname}%
 \mathlig@dolig#1\mathlig@delim
\fi
\mathlig@checksuffix#1#2\mathlig@end
}

\def\mathlig@checksuffix#1#2\mathlig@end{%
\ifx\mathlig@delim#2\mathlig@delim\relax\else\mathlig@checksuffix@{#1}#2\mathlig@end\fi
}
\def\mathlig@checksuffix@#1#2#3\mathlig@end{%
\expandafter\ifx\csname mathlig@forw@#1#2\endcsname\relax\mathlig@dosuffix{#1}{#2}\fi
\mathlig@checksuffix{#1#2}#3\mathlig@end
}

\def\mathlig@dosuffix#1#2{%
\mathlig@appendcs{mathlig@toks@#1}{#2}%
\mathlig@dolig{#1}{#2}\mathlig@delim
}

\def\mathlig@dolig#1#2\mathlig@delim{%
 \mathlig@defcs{mathlig@look@#1#2}{%
 \mathlig@let@cs\mathlig@next{mathlig@forw@#1#2}\futurelet\mathlig@next@tok\mathlig@next}%
 \mathlig@defcs{mathlig@forw@#1#2}{%
  \mathlig@let@cs\mathlig@next{mathlig@back@#1#2}%
  \mathlig@let@cs\checker{mathlig@chck@#1#2}%
  \mathlig@let@cs\mathligtoks{mathlig@toks@#1#2}%
  \expandafter\ifx\expandafter\mathlig@delim\mathligtoks\mathlig@delim\relax\else
  \expandafter\checker\mathligtoks\mathlig@delim\fi
  \mathlig@next
 }%
 \mathlig@defcs{mathlig@toks@#1#2}{}%
 \mathlig@defcs{mathlig@chck@#1#2}##1##2\mathlig@delim{%
  \ifx\mathlig@next@tok##1%
   \mathlig@let@cs\mathlig@next@cmd{mathlig@look@#1#2##1}\let\mathlig@next\mathlig@gobble
  \fi
  \ifx\mathlig@delim##2\mathlig@delim\relax\else
   \csname mathlig@chck@#1#2\endcsname##2\mathlig@delim
  \fi
 }%
 \ifx\mathlig@delim#2\mathlig@delim\else
  \mathlig@defcs{mathlig@back@#1#2}{\csname mathlig@back@#1\endcsname #2}%
 \fi
}%

\catcode`@\mathlig@atcode

\mathchardef\ordinarycolon\mathcode`\:
\def\vcentcolon{\mathrel{\mathop\ordinarycolon}}
\mathlig{:=}{\vcentcolon=}
\mathlig{::=}{\vcentcolon\vcentcolon=}

\numberwithin{equation}{section}

\theoremstyle{plain}
\newtheorem{theo}{Theorem}[section]
\newtheorem{cor}[theo]{Corollary}
\newtheorem{lem}[theo]{Lemma}
\newtheorem{prop}[theo]{Proposition}

\theoremstyle{definition}

\theoremstyle{remark}
\newtheorem{ex}{Example}
\theoremstyle{remark}
\newtheorem*{exs*}{Examples}
\theoremstyle{remark}
\newtheorem*{rems*}{Remarks}
\theoremstyle{remark}
\newtheorem*{rem*}{Remark}

\title[Cyclicity in rank-one perturbation problems]{Cyclicity in rank-one perturbation problems}
\author{Evgeny~Abakumov}
\thanks{The first author was partially supported by the ANR project DYNOP and the
ESF program "Harmonic and Complex Analysis and its
Applications"}
\address{Universit\'e Paris-Est, LAMA (UMR 8050), UPEMLV, UPEC, CNRS,
F-77454, Marne-la-Vall\'ee, France}
\email{evgueni.abakoumov@univ-mlv.fr}
\author{Constanze~Liaw}
\thanks{The second author is  partially supported by the NSF grant DMS-1101477.}
\address{Department of Mathematics, Texas A\&M University, Mailstop 3368,
 College Station, TX  77843, USA }
\email{conni@math.tamu.edu}
\urladdr{http://www.math.tamu.edu/\~{}conni}
\author{Alexei~Poltoratski}
\thanks{The third author is  partially supported by the NSF grant DMS-1101278.}
\address{Department of Mathematics, Texas A\&M University, Mailstop 3368,
 College Station, TX  77843, USA }
\email{alexeip@math.tamu.edu}
\urladdr{http://www.math.tamu.edu/\~{}alexei.poltoratski/}


\begin{document}

\begin{abstract}
The property of cyclicity of a linear operator, or equivalently the property of simplicity of its spectrum, is an important
spectral characteristic that appears in many problems of functional analysis and applications to mathematical physics.
In this paper we study cyclicity in the context of rank-one perturbation problems for self-adjoint and unitary operators.
We show that for a fixed non-zero vector the property of being a cyclic vector is not rare, in the sense that for
any family of rank-one perturbations of self-adjoint or unitary operators acting on the space, that vector will be
cyclic for every operator from the family, with a possible exception of a small set with respect to the parameter.
We  discuss applications of our results to Anderson-type Hamiltonians.

\end{abstract}

\maketitle

\

\

\section{Introduction}

\medskip\noindent Consider a self-adjoint operator $T$ on a separable Hilbert space $\cH$. A vector $\f\in \cH$ is called cyclic for an operator $T$, if
\[
\cH = \clos \spa\{(T-\lambda \OID)^{-1}\f:\lambda \in\C\backslash\R\}.
\]
An operator $T$ is called cyclic, if there exists a cyclic vector.
For a bounded operator $T$, an equivalent definition is  that
$$\cH = \clos \spa\{T^n\f:n\in \N\cup\{0\}\},$$
i.e., the span of the orbit of $\f$ under $T$ is dense in the Hilbert space.
Cyclicity of an operator is equivalent to the property that the operator has simple spectrum.
The property of simplicity of the spectrum often appears in problems originating from physics.

\medskip\noindent In this note we study cyclicity in the context of rank-one perturbation problems for self-adjoint and unitary
operators. If $A$ is a self-adjoint operator and  $\f$ is its cyclic vector one can consider the family of rank-one perturbations
\begin{align}\label{d-Rk1}
A_\alpha=A+\alpha(\fdot,\f)\f, \quad\text{for }\alpha\in\R.
\end{align}
Similar families can be defined for unitary operators, see Sections \ref{ss-SaRk1} and \ref{ss-ACT} for the definitions.
We show that the property of cyclicity for a fixed non-zero vector in the Hilbert space is not a rare event, in
the sense that for any family of cyclic rank-one perturbations a fixed vector is cyclic for \textit{all }operators in the family
with an exception of some small sets of parameters.

\medskip\noindent In the rank-one setting, for the singular part we prove that the exceptional set of parameters has Lebesgue measure zero; for the absolutely continuous part we prove "all but countably many"; and in some more special cases even "all but possibly one".

\medskip\noindent In particular, in Section \ref{s-CycAA} we prove that an arbitrary non-zero vector $\f$ in a Hilbert space is cyclic
for all but countably many operators $(A_\alpha)\ti{ac}$ for any family of self-adjoint (or unitary) rank-one perturbations
$A_\alpha$. An arbitrary non-zero vector is also cyclic for almost all operators $(A_\alpha)\ti{s}$ for any such family.
Here $(A_\alpha)\ti{ac}$ and $(A_\alpha)\ti{s}$ denote the absolutely continuous and singular parts of
the operators $A_\alpha$ respectively, see section
\ref{pre-normal} for the definitions.

\medskip\noindent In Subsection \ref{ss-ApplAnd} we discuss applications of our results to Anderson-type Hamiltonians and deduce and improve
some of the results from \cite{JakLast2000, JakLast2006} and \cite{Sim1994}.

\medskip\noindent In Section \ref{s-GDHE}, we show that if a vector belongs to a certain natural class of vectors, associated with the family $A_\alpha$, then we have cyclicity for all except possibly one operator $A_\alpha$.

\medskip\noindent The theory of rank-one perturbations of self-adjoint and unitary operators, and its applications to Anderson-type models became
an active area of research over the last 20 years. The interest to this part of perturbation theory is caused, to a large degree,
by connections to the famous problem of Anderson localization.

\medskip\noindent In 1958 P.~W.~Anderson (see \cite{And1958}) suggested that sufficiently large impurities in a semi-conductor could lead to spatial localization of electrons, called Anderson localization.
Although most physicists consider the problem solved, many mathematical questions with striking physical relevance remain open. The field has grown into a rich mathematical theory (see e.g.~\cite{Germ, Klo2, Kirsh} for different Anderson models and \cite{SIMSULE, ExSpec} for refined notions of Anderson localization).

\medskip\noindent While the property of localization for a random Anderson-type operator has many different definitions, one of the ``weaker" definitions of localization is equivalent to the property
that the spectrum of the operator is almost-surely singular. It is well-known that, if an Anderson-type Hamiltonian is almost-surely singular,
then it is almost-surely cyclic. Equivalently, if such an operator is not cyclic with positive probability, then it is delocalized.

\medskip\noindent The study of spectral behavior under rank-one perturbations proves to be one of the main tools in spectral analysis of
Anderson-type models, in particular in problems concerning cyclicity, see for instance \cite{JakLast2000, JakLast2006, Sim1994}. This connection
served as one of the motivating factors for the current paper.

\subsection{Notation/Beware}
We consider three main classes of families of perturbations: rank-one self-adjoint, rank-one unitary perturbations and Anderson-type Hamiltonians. Throughout we use the notations $A_\alpha$, $U_\gamma$
and $H_\omega$ to denote the corresponding families of perturbations, respectively. While all three classes are somewhat closely related, the two types of rank-one perturbations (\eqref{d-Rk1} and \eqref{d-ugamma} below) are almost interchangeable via the Cayley transform:
The Aleksandrov-Clark theory
and all its basic results discussed in Section \ref{ss-ACT} below can be equivalently re-stated in the case of the real line (upper half-plane).
For example, the Cauchy transform in $\D$, see the second equation of~\eqref{d-CTD}, is replaced with its analogue in $\C_+$, see equation~\eqref{e-F}
for the definition.
Similarly, results
on rank-one perturbations of self-adjoint operators can be re-formulated for the families of unitary rank-one perturbations, see for instance \cite{polt}. It is a well known feature of complex function theory that some of the proofs of the half-plane statements  look more natural in the settings of the unit disk and vice versa. Similarly, in this paper we  utilize both self-adjoint and unitary settings
in our statements and proofs.

\section{Preliminaries}\label{s-PRE}

\subsection{Cyclicity for normal operators}\label{pre-normal}
Recall that an operator in a separable Hilbert space is called normal if $T^*T= TT^*$. By the spectral theorem the operator $T$ is unitarily equivalent to $M_z$, multiplication by the independent variable $z$, in a direct sum of  Hilbert spaces
$$\cH = \oplus \int \cH(z)\, d\mu(z)$$ where $\mu$ is a scalar positive  measure  on $\C$. The measure $\mu$ is called a spectral measure
of $T$.

\medskip\noindent If $T$ is a unitary or self-adjoint operator, its spectral measure $\mu$ is supported on the unit circle or on the real line, respectively.
Via Radon decomposition, $\mu$ can be decomposed into a singular and an absolutely continuous parts $\mu=\mu\ti{s}+\mu\ti{ac}$.
The singular component $\mu\ti{s}$ can be further split into singular continuous and pure point parts.
For unitary or self-adjoint $T$ we denote by $T\ti{ac}$  the restriction of $T$ to its absolutely continuous part, i.e.~$T\ti{ac}$ is unitarily equivalent to $$M_t\big|_{\oplus \int \cH(t) d\mu\ti{ac}(t)}.$$ Similarly, define the singular, singular continuous and the pure point parts of  $T$, denoted by $T\ti{s}$, $T\ti{sc}$ and $T\ti{pp}$, respectively.

\medskip\noindent  In terms of the spectral representation described above, the property of cyclicity, as defined in the introduction,
 is equivalent to the property that  all Hilbert spaces $\cH(t)$ are one-dimensional and  the space
 $$\oplus \int \cH(t) d\mu(t)$$
 can be identified with $L^2(\mu)$.
 Cyclic vectors for $T$ correspond to functions with full support in  $L^2(\mu)$, i.e. those functions that are non-zero almost everywhere with respect to $\mu$.

\subsection{Self-adjoint rank-one perturbations}\label{ss-SaRk1}
Let $T$ be a normal operator on a Hilbert space $\cH$ and let $\f\in \cH$ be a non-zero vector.
An alternative definition of the spectral measure of $T$ can be given as follows. Notice that
 there exists a unique  measure $\mu$ on $\C$ such that
$$((T-\lambda\OID)^{-1} \f,\f)\ci{\cH} = \int \frac{d\mu(t)}{t-\lambda},$$
for all $\lambda$ outside of the spectrum of $T$. If $\mu$ is such a measure we say that $\mu$ is
the spectral measure of $T$ with respect to the vector $\f$. Note that such a measure is unique, once $T$ and $\f$ are fixed.
The operator $T$ is bounded if and only if $\mu$ is compactly supported.

\medskip\noindent Let $A$ be a self-adjoint operator and let $\f$ be its cyclic vector.
Consider the family of rank-one perturbations
\begin{align}\label{d-Rk1}
A_\alpha=A+\alpha(\fdot,\f)\f, \quad\text{for }\alpha\in\R.
\end{align}
It is not difficult to show that then $\f$ will be a cyclic vector for $A_\alpha$ for all $\alpha\in \R$.
Denote by $\mu_\alpha$  the spectral measure of $A_\alpha$ with respect to $\f$. In these notations $\mu=\mu_0$.

\medskip\noindent In virtue of the spectral theorem, one can always assume that $\cH=L^2(\mu)$, $A=M_t$ and $\f = \ID\in L^2(\mu)$.
Denote by $V_\alpha$ the operator of spectral representation for $A_\alpha$, i.e.   the unitary operator $V_\alpha:L^2(\mu)\to L^2(\mu_\alpha)$ such that $V_\alpha A_\alpha= M_t V_\alpha$ and $V_\alpha \f=\ID$.
An explicit formula for $V_\alpha$
was recently derived in \cite{mypaper}.

\medskip\noindent For unbounded $A$ (i.e.~not compactly supported $\mu$), we always assume that the spectral measure $\mu$ corresponding to $\f$ satisfies
$$\int\ci\R \frac{ d\mu(t)}{1+|t|}<\infty.$$
Using the standard terminology, this means that we consider the class of  singular form bounded perturbations and assume
that $\varphi\in\mathcal H_{-1}(A)\supset \cH$, i.e.~that $A\f\in \cH$. Notice that if $\f\notin\mathcal{H}_{-1}(A)$ then the formal expression \eqref{d-Rk1} does not possess a unique self-adjoint extension, see for instance
\cite{Kurasov}.

\medskip\noindent The Aronszajn--Donoghue theory analyzes the spectrum of the perturbed operator under rank-one perturbations. We will use the
following well-known statement:

\begin{theo}[\cite{SIMREV}]\label{t-AD}
For non-equal coupling constants $\alpha\neq\beta$, the singular parts $(\mu_\alpha)\ti{s}$ and $(\mu_\beta)\ti{s}$ are mutually singular.
\end{theo}

Of fundamental importance to many spectral problems is the Cauchy transform.
If $\tau$ is the spectral measure of a self-adjoint operator $A$ corresponding
to the vector $\f$ then the Cauchy transform of $\tau$,
\begin{align}\label{e-F}
\rK_\tau(z)=\frac{1}{\pi}\int_\R \frac{d\tau(t)}{t-z}\,,
\qquad
z\in \C\setminus \R,
\end{align}
 is equal  to the corresponding resolvent function of $A$:
 $$\rK_\tau(z)=((A-z \OID)^{-1}\f,\f)=\int \frac{d\mu(t)}{t-z}.$$
This connection allows one to apply complex analysis in spectral problems.

\medskip\noindent We use the notation
\[
(\rK_\tau)\ci+(x)=\lim_{y\to 0}\rK_\tau(x+ iy)
\qquad\text{and}\qquad
(\rK_\tau)\ci-(x)=\lim_{y\to 0}\rK_\tau(x- iy)
\]
for $x\in \R$.
By a theorem of Privalov,
\begin{align}\label{e-JUMP}
(\rK_\tau)\ci-(x)-(\rK_\tau)\ci+(x)=2\pi i \,\frac{d\tau}{dx}(x)
\end{align}
for Lebesgue-a.e. $x\in\R$.

\subsection{Aleksandrov--Clark theory and unitary rank-one perturbations}\label{ss-ACT}
By $H^p$ we will denote the standard Hardy spaces in the unit disk. Recall that a function $\te\in H^\infty$ is called inner, if $|\te(z)|=1$ for almost every $|z|=1$.  The (scalar valued) model space $K_\theta$ is defined as $K_\te=H^2\ominus \te H^2$. Such spaces play an important role in complex function theory and functional analysis, see for instance \cite{NikolskiTreat}.

\medskip\noindent The model operator on a space $K_\theta$ is defined as $S_\te=P_\te S$, where $P_\te$ denotes the orthogonal projection onto $K_\te$, while $S$ is the shift operator given by $Sf(z)=zf(z)$ for $f\in H^2$. The adjoint to the shift operator is the so-called
backward shift operator defined as
 $$S^*f=\frac{f(z)-f(0)}{z}\,.$$

\medskip\noindent Let $\theta$ be an inner function. To simplify the formulas we will assume that $\theta(0)=0$. In \cite{Clark} Clark showed that the family of rank-one perturbations
\begin{align}\label{unit}
\widetilde U_\gamma&=S_\te + \gamma (\fdot, S^*\te)\ID
\quad\text{for }\gamma\in\T
\end{align}
consists of unitary operators on $K_\te$, and - vice versa - that all unitary rank-one perturbations of the  model operator $S_\te= P_\te S |\ci{K_\te}$ are given by \eqref{unit}.

\medskip\noindent It is well-known that the  the vector $\ID\in K_\theta$ is cyclic for all operators $\widetilde  U_\gamma$ in the above family.
By $\sigma_\gamma$ denote the spectral measures of $\widetilde U_\gamma$ with respect to the function $\ID$, i.e.
such measures that the identity $$((\widetilde U_\gamma+z I)(\widetilde U_\gamma-z I)^{-1}\ID, \ID)= \int\ci\T \frac{\xi+z}{\xi-z} \,d\sigma_\gamma(\xi)$$ holds true for all $\gamma\in\T$. In fact, every inner function $\theta$ with $\theta(0)=0$ corresponds in a one-to-one fashion with a family of Clark measures $\{\sigma_\gamma\}_{\gamma\in\T}$. Further, for inner functions $\theta$, the measures $\sigma_\gamma$ are purely singular for all $\gamma\in \T$; and vice versa.

\medskip\noindent One of the main results of the Clark theory says that the spectral measures of $\widetilde U_\gamma$ are defined by the identity
\begin{align}
\label{thetamu}
 \frac{\theta+\gamma}{\theta-\gamma}
&=
\int\ci\T \frac{\xi+z}{\xi-z} \,d\sigma_\gamma(\xi).
\end{align}

\medskip\noindent The Clark operator
is the unitary operator $\Phi_\gamma: K_\te\to L^2(\sigma_\gamma)$ such that $\Phi_\gamma \widetilde U_\gamma= M_z
\Phi_\gamma$, where $M_z$ is the operator that acts as multiplication by the independent variable in
$L^2(\sigma_\gamma)$. In other words, the Clark operator is the spectral representation of $\widetilde U_\gamma$.

\medskip\noindent Notice that the spectral representation of $V_\alpha:L^2(\sigma)\to L^2(\sigma_\alpha)$ from the previous subsection on self-adjoint rank-one perturbations corresponds to the composition operator $\Phi_\gamma \Phi_{1}^{\ast}:L^2(\sigma_1)\to L^2(\sigma_\gamma)$ in the case of unitary rank-one perturbations.

\medskip\noindent The situation becomes more complicated without the assumption that the spectral measure is purely singular: In the case of non-trivial absolutely continuous spectrum the model space consists of pairs of functions analytic inside and outside of the unit disk, see
\cite{NikolskiII}.

\medskip\noindent However, many of the formulas of the Aleksandrov-Clark theory remain valid in the case where the spectral measures are not purely singular. If $\mu$ is a positive finite measure on the unit circle,
denote by $U_1$ the operator of multiplication by $z$ in $L^2(\mu)$. Let $\theta$ be a bounded holomorphic function in the unit disk $\D$
that satisfies \eqref{thetamu} for $\gamma=1$ and $\sigma_1=\mu$. If $\mu$ is not singular, $\theta$ is not inner but still
belongs to the unit ball of $H^\infty$, i.e. $|\theta|\leq 1$ in $\D$. Nonetheless, one
can still consider a family of measures $\sigma_\gamma$ defined by \eqref{thetamu}. This family will consist of spectral measures of unitary rank-one perturbations
of $U_1$ corresponding to the vector $\ID\in L^2(\mu)$, with $U_\gamma$  defined as
\begin{equation}\label{d-ugamma}
U_\gamma=U_1+(\gamma-1)(\cdot, U_1^*\ID)\ID,\ \gamma\in\T,
\end{equation}
see \cite{poltsara2006}.

\medskip\noindent We will also use the following two theorems from the Aleksandrov-Clark theory. Let $m$ denote the Lebesgue measure on $\T$.

\begin{theo}[Aleksandrov's spectral averaging, see e.g.~\cite{poltsara2006}]\label{t-ASA}
For $f\in L^1(\T, dm)$ we have
\[
\int f dm = \int \left(\int f d\sigma_\gamma \right)dm(\gamma).
\]
\end{theo}

\medskip\noindent It is well known that the adjoint $\Phi_\gamma^*: L^2(\sigma_\gamma)\to K_\te$ of the Clark operator can be represented using the normalized Cauchy transform
\[
\Phi_\gamma^* h= \frac{\cK_{h\sigma_\gamma}}{\cK_{\sigma_\gamma}}\,,\]
where $\cK$ stands for the Cauchy transform in $\D$:
\begin{equation}\label{d-CTD}
\cK_{\sigma_\gamma}(z)=\int\ci\T \frac{d\sigma_\gamma(\xi)}{1-\bar\xi z}\,,
\qquad\textrm{and}\qquad
\cK_{h\sigma_\gamma}(z)=\int\ci\T \frac{h(\xi)d\sigma_\gamma(\xi)}{1-\bar\xi z}\,.
\end{equation}

\begin{theo}[\cite{NONTAN}]\label{t-P}
For any $f\in L^1(\sigma_\gamma)$,
\[
\lim_{r\to 1}\frac{\cK_{f\sigma_\gamma}(rz)}{\cK_{\sigma_\gamma}(rz)}=f(z),\textrm{ for }
(\sigma_\gamma)\ti{s} \text{-a.e.}\ z\in\T.
\]
\end{theo}

\medskip\noindent Theorems \ref{t-ASA} and \ref{t-P}, hold true even if the family of spectral measures possesses non-trivial absolutely continuous parts, although the normalized Cauchy transform of an arbitrary function in $L^2(\sigma_\gamma)$ cannot be interpreted as a vector from $K_\theta$.

\medskip\noindent Via the standard agreement, every function from a Hardy space $H^p$ is identified with its boundary values on the circle $\T$. It follows from Theorem \ref{t-P} that the boundary values of any $f\in K_\theta$ exist $\sigma_\gamma$-a.e for any $\gamma\in\T$ and the Clark operator
$\Phi_\gamma: K_\theta\to L^2(\sigma_\gamma)$ simply sends $f$ into its boundary values.

\medskip\noindent For a function $f\in K_\te$, denote  by $\widetilde f$ the function $\te \bar f$ on $\T$. Note that $f\in K_\theta$ and $f(0)=0$ imply that $\theta \bar f\in K_\theta$.
A function $f\in K_\te$ is called a Hermitian element, if $\widetilde f= f$. Notice that Hermitian functions satisfy $f(0)=0$.

\medskip\noindent The following simple statement plays an important role in Section \ref{s-GDHE}.

\begin{theo}[\cite{NONTAN}]\label{t-Hermite}
Let $f\in K_\te$. Then $f$ is a Hermitian element if and only if
\begin{equation}
\arg(\Phi_\gamma f)=\frac{\arg\gamma}{2} (\Mod \pi),\ \sigma_\gamma\text{-a.e.},
\label{eq01}\end{equation}
and  $\int fd\sigma_\gamma=0$ for some $\gamma\in\T$. If $f$ is a Hermitian element then $f$ satisfies \eqref{eq01} and
$\int fd\sigma_\gamma=0$ for any $\gamma\in\T$.

\end{theo}

\subsection{Spaces of Paley--Wiener functions}\label{ss-PAL}
For $a>0$ the class of Paley--Wiener functions on $\R$ is given by
\[
 \PW_a = \{\hat f:f\in L^2(-a,a)\},
\]
where $\hat f(z) = \int e^{-izt} f(t) dt$ denotes the classical Fourier transform of $f$. Alternatively, the Paley--Wiener space can be characterized as the space of entire functions of exponential type at most $a$ whose boundary values on the real line  are square summable with respect to Lebesgue measure.

\medskip\noindent  The Paley--Wiener space $\PW_a$ is closely related to the model space $K_\te$ for the inner function
 $$\theta_a(z) = \theta(z) = e^{-2a\frac{1+z}{1-z}}$$
  in the unit disk. To establish the connection, consider the conformal map $$\psi(z)=\frac{z-i}{z+i}$$ from $\C_+$ to $\D$.
  Denote $\vartheta(z)=\vartheta_a(z) = e^{2iaz}$. Note that
  $$\vartheta_a(z)=\theta_a(\psi(z)).$$
  By $K_\vartheta^\R$ denote the space obtained from $K_\te$
  by composing all functions from $K_\te$ with $\psi$, i.e.
  $$K_\vartheta^\R=\{f(\psi)|f\in K_\te\}.$$
Then the space
$$e^{-iaz}K_\vartheta^\R=\{e^{-iaz}f(\psi)|f\in K_\vartheta\}$$
 is equal to the space of entire functions of exponential type at most $a$ and with boundary values on $\R$ that are square summable with respect to the measure $(1+x^2)^{-1} dx$.

\medskip\noindent Hence we have
\[
 \PW_a \subset e^{-iaz}K_{\vartheta_a}^\R
\qquad \text{for }0<a.
\]
Further, one can prove that the codimension is equal to 1 and $$e^{-iaz}K_{\vartheta_a}^\R\ominus \PW_a$$ consists of constant functions.

\section{Arbitrary non-zero vectors yield cyclic vectors for almost all parameters}\label{s-CycAA}
\medskip\noindent Let $A$ be a self-adjoint (possibly unbounded) operator on a separable Hilbert space $\cH$ and let $\f$ be a cyclic vector for $A$.
Define the family of self-adjoint rank-one perturbations of $A$, $A_\alpha$ as in Subsection \ref{ss-SaRk1}.
Recall that $(A_\alpha)\ti{ac}$ and $(A_\alpha)\ti{s}$ denote the absolutely continuous part and the singular part of the operator $A_\alpha$, respectively.

\medskip\noindent Further notice that $A_\alpha=(A_\alpha)\ti{s}\oplus (A_\alpha)\ti{ac}$, since $(\mu_\alpha)\ti{s}\perp (\mu_\alpha)\ti{ac}$.

\begin{theo}\label{t-GENCyc}
Let $A_\alpha$ be a family of self-adjoint rank-one perturbations in a Hilbert space $\cH$ given by \eqref{d-Rk1}. Let $0\neq f\in \cH$. Then
\begin{itemize}
\item[1)] The function $f$ is a cyclic vector for $( A_\alpha)\ti{ac}$ for all but a countable number of $\alpha\in\R$.
\item[2)] The function $f$ is a cyclic vector for $( A_\alpha)\ti{s}$ for Lebesgue a.e.~$\alpha\in\R$.
\end{itemize}
\end{theo}

\medskip\noindent In the following subsection we prove an equivalent reformulation of this theorem in terms of its  spectral representation.
\subsection{Proof of Theorem \ref{t-GENCyc}}
As we mentioned before, in view of the spectral theorem, instead of dealing with a general family of self-adjoint rank-one perturbations given by equation \eqref{d-Rk1} one can consider the self-adjoint rank-one perturbations $A_\alpha = M_t + \alpha (\fdot, \ID)\ci{L^2(\mu)}\ID$ on $L^2(\mu)$.  Let the spectral operator $V_\alpha: L^2(\mu)\to L^2(\mu_\alpha)$ be as defined in Subsection \ref{ss-SaRk1}. In these settings Theorem \ref{t-GENCyc} can be stated as follows.

\begin{theo}\label{t-RK1}
Let $0\neq f\in L^2(\mu)$. Then
\begin{itemize}
\item[1)] The function $f_\alpha=V_\alpha f\in L^2(\mu_\alpha)$ is not equal to zero $(\mu_\alpha)\ti{ac}$-a.e. for all but a countable number of $\alpha\in\R$.
 \item[2)] The function $f_\alpha=V_\alpha f\in L^2(\mu_\alpha)$ is not equal to zero $(\mu_\alpha)\ti{s}$-a.e.  for Lebesgue a.e.~$\alpha\in\R$.
\end{itemize}
\end{theo}

\begin{rems*}
(a) Obviously one cannot expect to obtain the above conclusion of cyclicity for all $\alpha\in \R$. It is always possible to start with  $f$ that is zero on a set of positive $\mu_0$-measure that is, therefore, not cyclic for $\alpha=0$. In fact, under the conditions of Theorem \ref{t-RK1}, we cannot replace
"countable" with "finite."\\
(b) In the case of purely singular spectral measures for some natural classes of $f$ the conclusion can be strengthened  to "all but one" $\alpha$, see Section \ref{s-GDHE} below.
\end{rems*}

\medskip\noindent Our next example  says that, in general, if $f$ is not Hermitian, then $f$ can be non-cyclic for uncountably many corresponding rank-one perturbations. In particular, in the conclusion of Theorem \ref{t-RK1} the distinction between the singular and the absolutely continuous is necessary.
As was mentioned above, throughout the rest of the paper we will switch between self-adjoint and unitary settings
as a matter of convenience. An analogous discussion can always be carried out in the other case.

\begin{ex}\label{exa-two}
Consider the setting of rank-one unitary perturbation described in Subsection \ref{ss-ACT}. We will construct a bounded holomorphic function $\theta$ such that for
the family of spectral measures $(\sigma_\gamma)_{\gamma \in \T}$ defined by  \eqref{thetamu} and the corresponding family of unitary operators $U_\gamma$ there exists a non-zero function $f\in L^2(\sigma_1)$ such that $f_\gamma= \Phi_\gamma \Phi_{1}^{\ast} f$ is non-cyclic in $ L^2(\sigma_\gamma)$ for an uncountable set of $\gamma$'s. Note that in this example $\theta$ is not inner and the corresponding operators have nontrivial absolutely continuous parts.

\medskip\noindent Let $C$ be a Cantor (closed uncountable) subset of the unit circle $\T$. Let $w $ be the continuous function on $\T$ defined by $w(\xi) = dist^2(\xi, C)$.  Denote by $(\sigma_\gamma)_{\gamma \in \T}$ the system of probability measures which is the family of Clark measures for some inner function $\theta$, and such that the measure $\sigma_1$ coincides, up to a multiplicative constant, with the measure $ w\,dm$, where $m$ is the Lebesgue measure on $\T$. Note that, by definition of $\sigma_1$, we have
$$\int\frac{1}{|x-y|^2}\,d\sigma_1(y)<\infty$$
for any $x\in C$. It is well known (see, e.g.~\cite{cimaross}) that
the last condition implies that each point of $C$ is a point mass for one of the measures $\sigma_\gamma$. Since $C$ is uncountable,
we conclude that
uncountably many $\sigma_\gamma$'s have atoms on the set $C$.

\medskip\noindent Let now $F$ be an outer function with modulus equal to $w$ almost everywhere on $\T$. Consider $f=F/w$. Then $f$ is a unimodular function on $\T$, and we have $ f\sigma_1 = fw=F$, hence $\cK _{ f\sigma_1} = 0$ on the set $C$. So we have
$$\frac{\cK _{ f\sigma_1}}{\cK _{ \sigma_1}} = 0\qquad\text{on } C.$$ In virtue of the following Lemma \ref{l-resolvent1}, it follows that $ \Phi_\gamma \Phi_{1}^{\ast} f=0$ on the (uncountable) set of those $\gamma \in \T$ for which $\sigma_\gamma$ has a point mass on $C$.
Hence $f$ is not cyclic for uncountably many operators $U_\gamma$.
\end{ex}

\bigskip

\begin{lem}[Aronszajn--Krein-type formula]\label{l-resolvent1}
Under the hypotheses of Theorem \ref{t-RK1} we have
\begin{align}\label{e-AK}
\rK_{f_\alpha\mu_\alpha}=\frac{\rK_{f_{\beta}\mu_{\beta}}}{1+(\alpha-\beta) \rK_{\mu_{\beta}}}
\qquad\text{and}\qquad
\frac{\rK_{f_\alpha\mu_\alpha}}{\rK_{\mu_\alpha}}=\frac{\rK_{f_0\mu_0}}{\rK_{\mu_0}}\,.
\end{align}
\end{lem}

\begin{proof}[Proof of Lemma \ref{l-resolvent1}]
First consider the case where $\ID\in\cH=L^2(\mu)$. Let $z\in \C\backslash \R$. Combining the second resolvent equation and the fact that $A_\alpha = A_\beta +(\alpha - \beta) (\fdot, \ID)\ID$ we obtain
\[
(A_{\beta}-z\OID)^{-1}\ID - (A_\alpha-z\OID)^{-1}\ID = (\alpha-\beta) ((A_\alpha-z\OID)^{-1} \fdot, \ID) (A_{\beta}-z\OID)^{-1}\ID.
\]
Application to a vector $f\in\cH$ and pairing with $\ID$ yields
\[
((A_{\beta}-z\OID)^{-1}f,\ID) - ((A_\alpha-z\OID)^{-1}f,\ID) = (\alpha-\beta) ((A_\alpha-z\OID)^{-1} f, \ID) ((A_{\beta}-z\OID)^{-1}\ID,\ID).
\]

\medskip\noindent Recall that $V_\alpha A_\alpha= M_t V_\alpha$, $V_\alpha \ID=\ID$ and $V_\alpha f=f_\alpha$. With this we obtain
\[
\rK_{f_{\beta}\mu_{\beta}}-\rK_{f_\alpha\mu_\alpha}=(\alpha-\beta) \rK_{f_\alpha\mu_\alpha} \rK_{\mu_{\beta}},
\]
or equivalently,
\begin{align}\label{e-AK0}
\rK_{f_\alpha\mu_\alpha}=\frac{\rK_{f_{\beta}\mu_{\beta}}}{1+(\alpha-\beta) \rK_{\mu_{\beta}}}= \frac{\rK_{f_{\beta}\mu_{\beta}}}{\rK_{\mu_{\beta}}}\,\frac{\rK_{\mu_{\beta}}}{1+(\alpha-\beta) \rK_{\mu_{\beta}}}=\frac{\rK_{f_{\beta}\mu_{\beta}}}{\rK_{\mu_{\beta}}}\,\rK_{\mu_\alpha}.
\end{align}
In the last equality we used the well-known Aronszajn--Krein formula, which can be obtained from the first equality of \eqref{e-AK0} by using $f=\f$ (or equivalently $f_\alpha=\ID$).

\medskip\noindent To obtain the second formula of \eqref{e-AK}, we divide both sides by $\rK_{\mu_\alpha}$.

\medskip\noindent If $\ID\in\cH_{-1}(A)\backslash \cH$ the resolvent formula is slightly more complicated
\begin{align*}
(A_\alpha-\lambda\OID)^{-1}f
&=(A_\beta-\lambda\OID)^{-1}f-\frac{(\alpha-\beta)\left((A_{\beta}-\lambda\OID)^{-1}f,\ID\right)}{1+(\alpha-\beta)\left((A_{\beta}-\lambda\OID)^{-1}\ID,\ID\right)}(A_{\beta}-\lambda\OID)^{-1}\ID
\end{align*}
for $f\in\cH_{-1}(A)$, see e.g.~\cite{KL}. When paired with the vector $\ID$ this yields
\begin{align*}
\rK_{f_\alpha\mu_\alpha}=\left(1-\frac{(\alpha-\beta) \rK_{\mu_{\beta}}}{1+(\alpha-\beta) \rK_{\mu_{\beta}}}\right)\rK_{f_{\beta}\mu_{\beta}}
= \frac{\rK_{f_{\beta}\mu_{\beta}}}{1+(\alpha-\beta) \rK_{\mu_{\beta}}}.
\end{align*}

\medskip\noindent The remainder of the proof for $\ID\in\cH_{-1}(A)\backslash \cH$ now follows similarly to the case of regular perturbations $\ID\in\cH$.
\end{proof}

\bigskip
\begin{proof}[Proof of part 1) of Theorem \ref{t-RK1}]
Define the set $$\sa=\{x\in\supp(\mu_\alpha)\ti{ac}:f_\alpha(x)=0\}.$$ The goal is to show that $(\mu_\alpha)\ti{ac}(\sa)=0$ for all but a countable number of parameters $\alpha$.

\medskip\noindent Assume that $f_\alpha$ is not cyclic for uncountably many $\alpha\in\R$, i.e.~assume that for some $S\subset\R$, $S$ uncountable, we have $(\mu_\alpha)\ti{ac}(\sa)>0$ for all $\alpha \in S$. Then $|\sa|>0$ for all $\alpha\in S$. Since $S$ is uncountable
\begin{align}\label{e-sa}
|\sa\cap\Sigma_\beta|>0
\qquad\text{for some }\alpha,\beta\in S\text{ with }\alpha\neq\beta.
\end{align}

\medskip\noindent  Let us fix $\alpha$ and $\beta$ satisfying \eqref{e-sa} and investigate the jump behavior in the first equation of \eqref{e-AK} below. By Fatou's jump theorem, see equation \eqref{e-JUMP}, we have
\[
(\rK_{f_\alpha\mu_\alpha})\ci-(x)-(\rK_{f_\alpha\mu_\alpha})\ci+(x)=2\pi i \,\frac{d(f_\alpha\mu_\alpha)}{dx}(x)
=0
\qquad\text{Lebesgue-a.e.~}x\in(\sa\cap\Sigma_\beta),
\]
because $f_\alpha=0$ on $\sa$.

\medskip\noindent Similarly $\rK_{f_\beta\mu_\beta}$ has no jump Lebesgue almost everywhere on $\sa\cap\Sigma_\beta$. On the other hand $\rK_{\mu_\beta}$  has a non-zero jump Lebesgue almost everywhere on $(\sa\cap\Sigma_\beta)\subset\supp(\mu_\alpha)\ti{ac}$.

\medskip\noindent Hence, while the right hand side in the first equation of \eqref{e-AK} has a jump, the left hand side does not jump Lebesgue-a.e.~on $\sa\cap\Sigma_\beta$ where we have \eqref{e-sa}, and we arrive at a contradiction.
Therefore the assumption that $f_\alpha$ is not cyclic for uncountably many $\alpha\in\R$ cannot be maintained.
\end{proof}

\begin{proof}[Proof of part 2) of Theorem \ref{t-RK1}]
Denote by $S_\alpha$ the essential support of the measure  $(\mu_\alpha)\ti{s}$ defined as the
set of points where the Radon derivative of  $(\mu_\alpha)\ti{s}$ is infinite.
It follows from the the Aronszajn--Donoghue theorem, Theorem \ref{t-AD},
that the sets $S_\alpha$ are disjoint.
Define the set $$\oa=\{x\in S_\alpha :f_\alpha(x)=0\}.$$ The goal is to show that $(\mu_\alpha)\ti{s}(\oa)=0$ for Lebesgue a.e.~$\alpha\in \T$.

\medskip\noindent Assume that $f_\alpha$ is not cyclic for a set of $\alpha$'s with positive Lebesgue measure, i.e.~assume that for some $S\subset\R$, $|S|>0$ we have
\[
\mu_\alpha(\oa)>0\qquad\text{for all }\alpha \in S.
\]

\medskip\noindent In virtue of Theorem \ref{t-P}, (or, more precisely, its analog for the real line) we have for the boundary values
\begin{align}\label{e-to0}
\frac{\rK_{f_\alpha\mu_\alpha}(x+iy)}{\rK_{\mu_\alpha}(x+iy)}\,\,\,
\stackrel{y\to 0}{\longrightarrow} \,\,0
\end{align}
for  $(\mu_\alpha)\ti{s}$-a.e. $x\in \oa$ and all $\alpha \in S$.
Therefore, there exists a set $M\subset\R$, $|M|>0$ such that for all $x\in M$ there exists an $\alpha$ such that \eqref{e-to0} is satisfied.

\medskip\noindent Using Lemma \ref{l-resolvent1} (below), the analytic function $$\frac{\rK_{f_\alpha\mu_\alpha}}{\rK_{\mu_\alpha}}=\frac{\rK_{f_0\mu_0}}{\rK_{\mu_0}}$$ has zero boundary values on $M$, a set of Lebesgue measure greater than zero.

\medskip\noindent Hence $$\frac{\rK_{f_0\mu_0}}{\rK_{\mu_0}}\equiv 0$$ and we must have $f_0\equiv 0$. But this contradicts the hypothesis that $f$ is a non-zero vector.
\end{proof}

\medskip\noindent Let us prove the lemma that was used in the above proofs. Recall the definition \eqref{e-F} of the Cauchy transform.

\subsection{A Corollary of Lemma \ref{l-resolvent1}}
Let us mention another consequence of Lemma \ref{l-resolvent1}, although we will not use this fact later in this paper. Consider a family of Aleksandrov--Clark measures $\{\sigma_\gamma\}_{\gamma\in\T}$.
\begin{cor}
 We have
\[
 \frac{\cK_{f_\gamma\sigma_\gamma}}{\cK_{f_1\sigma_1}}(z) = \frac{1-\te(z)}{1-\bar\gamma\te(z)}
\qquad  m-\text{a.e.~}z\in\C\backslash\T.
\]
In particular, the function
\begin{align}\label{e-inline}
\frac{\cK_{f_\gamma\sigma_\gamma}}{\cK_{f\sigma_1}}
\end{align}
is independent of the choice of $f = f_1\in L^2(\sigma_1)$.
\end{cor}

\begin{rem*}
 It was R.~G.~Douglas who observed the independence of the expression \eqref{e-inline} from the choice of $f\in L^2(\sigma_1)$.
\end{rem*}

\begin{proof}
In order to see the second statement notice that by Lemma \ref{l-resolvent1}, the expression \eqref{e-inline} is independent from the choice of $f\in L^2(\sigma_1)$, i.e.
\[
 \frac{\cK_{f_\gamma\sigma_\gamma}}{\cK_{f\sigma_1}}=\frac{\cK_{g_\gamma\sigma_\gamma}}{\cK_{g\sigma_1}}
\qquad\text{for all }f,g\in L^2(\sigma_1).
\]

\medskip\noindent We obtain the first statement by expanding
\[
 \frac{\cK_{f_\gamma\sigma_\gamma}}{\cK_{f\sigma_1}} = \frac{\frac{\cK_{f_\gamma\sigma_\gamma}}{\cK_{\sigma_\gamma}} \cK_{\sigma_\gamma}}{\frac{\cK_{f\sigma_1}}{\cK_{\sigma_1}}\cK_{\sigma_1}}
\]
and use Lemma \ref{l-resolvent1} for $g=\ID$ to cancel the fractions in the numerator and denominator. Further apply equation $$\cK_{\sigma_\gamma}(z) = \frac{1}{1-\bar\gamma \te(z)}$$ (confer of \cite{cimaross}) to the remaining $\cK_{\sigma_\gamma}$ in the numerator as well as $\cK_{\sigma_1}$ in the denominator.
\end{proof}

\subsection{Anderson-type Hamiltonians}\label{ss-Hw}

\medskip\noindent The following operator is a generalization of most Anderson models discussed in literature.

\medskip\noindent For $n=1,2,...$ consider the probability space $\Omega_n=(\R, \mathcal B, \mu_n)$, where $\mathcal B$ is the Borel sigma-algebra on $\R$ and $\mu_n$
is a Borel probability measure on $\Bbb R$. Let $\Omega=\prod_{n=1}^\infty \Omega_n$ be a product space with the probability measure $\p$ on
$\Omega$  introduced as the product measure of the corresponding measures on $\Omega_n$ on the product sigma-algebra $\mathcal A$. The elements of $\Omega$ are points in
$\R^\infty$,  $\omega=(\omega_1,\omega_2,...), \omega_n\in\Omega_n$.

\medskip\noindent Let $\cH$ be a separable Hilbert space. Consider a self-adjoint operator $H$ on $\cH$ and let $\f_1,\f_2,\hdots$ be a countable
collection of non-zero vectors in $\cH$. For each $\omega\in\Omega$ define an Anderson-type Hamiltonian on $\cH$ as a self-adjoint operator formally given by
\begin{equation}\label{Model}
H_\omega = H + V_\omega, \qquad V_\omega = \sum\limits_n \omega_n (\fdot, \f_n)\f_n.
\end{equation}

\medskip\noindent We will suppose that the operator $H_\omega$ is densely defined $\p$-almost surely. Except for degenerate cases, the perturbation $V_\omega$ is almost-surely a non-compact operator. It is hence not possible to apply results from classical perturbation theory to study the spectra of $H_\omega$, see e.g.~\cite{birst} and \cite{katobook}.

\medskip\noindent In the case of an orthonormal sequence $\{\f_n\}$, this operator was studied in \cite{JakLast2000} and \cite{JakLast2006}.

\medskip\noindent Probably the most important special case of an Anderson-type Hamiltonian is the discrete random Schr\"odinger operator on $l^2(\Z^d)$
\[
Hf(x)=-\bigtriangleup f (x) = - \sum\limits_{|n|=1} (f(x+n)-f(x)), \quad \f_n(x)=\delta_n(x)=
\left\{\begin{array}{ll}1&\text {if}     \,\,x=n\in\Z^d,\\ 0&\text{else.}\end{array}\right.
\]

\subsection{An application of Theorem \ref{t-GENCyc} to Anderson-type Hamiltonians}\label{ss-ApplAnd}
Let $H_\omega$ be the Anderson-type Hamiltonian introduced in equation \eqref{Model}. Fix $\omega_0\in \Omega$. Assume $\f\in \cH_{-1}(H_{\omega_0})$ is a cyclic vector for the self-adjoint operator $H_{\omega_0}$. Consider operators $H_{\omega_0} + \alpha (\fdot, \f)\f$, $\alpha\in \R$.

\medskip\noindent Then (by Theorem \ref{t-GENCyc}) any non-zero $f\in \cH$ is cyclic for $H_{\omega_0} + \alpha (\fdot, \f)\f$ for almost all $\alpha\in\R$.
In particular, for Lebesgue almost every $\alpha$, the operators $H_{\omega_0} + \alpha (\fdot, \f)\f$ are cyclic.

\medskip\noindent In the case where $\f\in\clos\spa\{\f_n\},\ \f=\sum a_n \f_n$, we say that $\f$ corresponds to the (possibly non-unique) sequence $\textbf{a} = (a_1,a_2,\hdots)$. Further, the operators $H_{\omega_0} + \alpha (\fdot, \f)\f$ correspond to $\omega$ belonging to the one dimensional affine subspace
$$l(\omega_0,\mathbf{a})=\{\omega_0+\alpha (a_1, a_2, a_3, \hdots)|\ \alpha\in\R\}.$$
Cyclicity of the operators for almost every $\omega$ in any one-dimensional affine subspace is a stronger statement than  $\p$-almost-sure cyclicity that can be found in the literature for some particular cases of our model. In terms of almost-sure cyclicity we obtain the following
result.

\medskip\noindent If $l(\omega_0,\mathbf{a})$ is a one-dimensional affine subspace of $\R^\infty$, one can introduce Lebesgue measure
on $l$ as $$m(S)=|\{\alpha|\ \omega_0+\alpha (a_1, a_2, a_3, \hdots)\in S\}|$$ for any Borel subset $S$ of $l$.


\begin{cor}\label{c-AndHam}
 Suppose that $H_{\omega_0}$ is self-adjoint for some $\omega_0$ and that $\sum a_n \f_n$, $\mathbf{a} = (a_1,a_2, \hdots)$, is cyclic for $H_{\omega_0}$.
Consider  a one-dimensional affine subspace of $\R^\infty$, $l=l(\omega_0,\mathbf{a})$.
Then  any non-zero vector $\f$ is cyclic for all $(H_\omega)\ti{ac},\ \omega\in l$, except possibly countably many $\omega$, and
cyclic for almost every $(H_\omega)\ti{s},\ \omega\in l$, with respect to Lebesgue measure on $l$.

\medskip\noindent In particular, suppose that the probability measure $\p$ is a product of absolutely continuous measures,  $H_\omega$ is self-adjoint $\p$-almost surely and some $\f_n$ is cyclic for  $H_\omega$, $\p$-almost surely. Then any non-zero $\f\in \cH$ is cyclic for $H_\omega$, $\p$-almost surely.
\end{cor}

\medskip\noindent It is well-known that if an Anderson-type Hamiltonian is singular almost-surely then it is cyclic almost-surely. The proof of almost-sure cyclicity of the singular part $(H_\omega)\ti{s}$ and almost-sure cyclicity of certain specific vectors can be found in \cite{JakLast2006} and for the discrete Schr\"odinger operator in \cite{Sim1994}. The second part of Corollary \ref{c-AndHam} extends (from the singular part to the full Anderson-type Hamiltonian $H_\omega$) these results showing that if one of the vectors $\f_n$ is almost sure cyclic then any non-zero vector possesses that property.

\begin{proof}[Proof of Corollary \ref{c-AndHam}]
The first statement follows immediately from Theorem \ref{t-GENCyc}.

\medskip\noindent Let $\omega$ be such that  $\f_n$ is a cyclic vector for $H_{\omega}$.
For $\mathbf{a}=(0,\hdots,0,a_n,0,\hdots)\in \R^\infty$ every $0\neq\f\in \cH$ is cyclic for a.e.~point in $l(\omega, \mathbf{a})$. Since the union of such
subspaces covers $\p$-almost all points of $\R^\infty$, we obtain the statement.
\end{proof}

\section{General Hermitian elements and rank-one perturbations}\label{s-GDHE}
\medskip\noindent Let $U$ be a unitary operator. For a vector $\f$ consider the space $X$ defined as the closure of the set of real finite linear combinations of elements of the form $$(U+U^*)^n \f\qquad \text{and}\qquad\frac{1}{i}(U-U^*)^n\f\qquad \text{for} \qquad n\in\Z.$$ Then a vector $f\in\cH$ is \emph{Hermitian} with respect to $U$ and the vector $\f$, if $f\in X$ and $f\perp\f$.

\medskip\noindent An analogous definition can be given for self-adjoint operators. For a bounded self-adjoint operator $A$ on a separable Hilbert space $\cH$ and a vector $\f$, let $(\re A) \f$ denote the closure of the space of linear combinations  of $A^n \f$, $n\in\N$ with real coefficients.
We say that a vector $f\in\cH$ is \emph{Hermitian} with respect to the operator $A$ and the vector $\f$, if $f\in (\re A)\f$ and $f\perp \f$.
For general (unbounded) operators $(\re A) \f$ can be defined as the closed span of
$$\left((A-zI)^{-1}+(A-\bar zI)^{-1}\right)\f,\ z\in \C_+.$$

\medskip\noindent Note that in the settings of Section \ref{ss-ACT}, the space $X$ defined above is the
set of Hermitian functions from $K_\theta$ defined there, see the proof of Theorem \ref{t-herm2} below.

\medskip\noindent Let $U$ be a unitary operator on a separable Hilbert space $\cH$. Consider the family
\begin{align}\label{f-Ugamma}
U_\gamma=U+(\gamma-1) ( \fdot, U^{-1}  b )\ci{\cH}  b
\end{align}
of rank-one perturbations, $\gamma\in \T$, $ b \in\cH$ with $\| b \|\ci\cH =1$. It is well known that $U_\gamma$ is unitary for all $\gamma\in \T$. Clearly we have $U=U_1$. Without loss of generality, assume that $ b $ is cyclic for $U$, i.e.
$$\clos\spa\{U^k  b :k\in\Z\}=\cH.$$

\begin{theo}\label{t-herm2}
Consider the family $U_\gamma$ of rank-one unitary perturbations given by \eqref{f-Ugamma}. Assume that $U_\gamma$  has purely singular spectrum
for some (all) $\gamma\in \T$. Let $0\neq f\in \cH$ be Hermitian with respect to $U=U_1$ and $b$. Fix a constant $c\in \C\backslash\{0\}$. Then the vector $f-cb$ is cyclic for $U_\gamma$ for all $\gamma\in \T\backslash\{e^{2i\arg c}\}$.
\end{theo}

\subsection{Proof of Theorem \ref{t-herm2}}\label{ss-proof} Via the spectral theorem, without loss of generality one can assume that $U=U_1$
is an operator of multiplication by $z$ in $L^2(\sigma_1)$, $\sigma_1$ is the spectral measure of $U$ corresponding to $b$ and $b=\ID\in L^2(\sigma_1)$.
Define the inner function $\theta$ so that $\sigma_1$ is its Clark measure.
Let  $K_\theta$, $\gamma$, $\sigma_\gamma$ as well as $\Phi_\gamma$ and $\Phi^\ast_\gamma$ be as  defined in Subsection \ref{ss-ACT}.
Consider operators
$\widetilde U_\gamma$ defined in \eqref{unit}. Then $U_\gamma=\widetilde U_\gamma$.

\medskip\noindent Recall that $f\in K_\theta$ is called Hermitian if $f=\widetilde f$, where $\widetilde f = \theta \bar f$.

\medskip\noindent Using Theorem \ref{t-Hermite} one can show that this definition is equivalent to that of a Hermitian element with respect to operator $U=\widetilde U_1$ and vector $b=S^*\theta$. The condition $f(0)= 0$ (for $f\in K_\te$) is translated  into $f\perp b$.

\medskip\noindent Recall that by Theorem \ref{t-P}, the non-tangential limit of $f\in K_\theta$ exist $\sigma_\gamma$-a.e.~for all $\gamma\in\T$. Let us denote this non-tangential limit by $$f_\gamma(z)=\lim_{\xi\to z} f(\xi) \qquad \sigma_\gamma\text{-a.e.}.$$ In fact, we can identify these boundary values with one function, say $f$ (slightly abusing notation), on the circle which is defined $\sigma_\gamma$-a.e.~for all $\gamma$. Indeed, we restricted ourselves to purely singular measures and by the Aronszajn--Donoghue theorem, Theorem \ref{t-AD}, we have $\sigma_\gamma\perp\sigma_\eta$ for $\gamma\neq\eta$. Our statement now follows from
Theorem \ref{t-herm1} below.

\begin{theo}\label{t-herm1}
Let $0\neq f\in K_\theta$ be a Hermitian function and fix any constant $c\in \C\backslash\{0\}$. Then the level sets $$\{z\in\T : f(z)=c\}$$ have zero $\sigma_\gamma$-measure for all $\gamma\in \T\backslash\{e^{2i\arg c}\}$.
In particular, the function $f-c$ is cyclic for $\widetilde U_\gamma$ for all $\gamma\in \T\backslash\{e^{2i\arg c}\}$.
\end{theo}

\begin{proof}[Proof of Theorem \ref{t-herm1}]
Pick $f$ and $c$ according to the hypotheses of the theorem.
By Theorem \ref{t-Hermite} we have $$\arg f = \frac{\arg\gamma}{2}(\Mod \pi)$$ with respect to $\sigma_\gamma$ almost everywhere.

\medskip\noindent If $\gamma$ is such that $f=c$  on a set $S\subset\T$ with $\sigma_\gamma(S)>0$, then we have $$\frac{\arg\gamma}{2} (\Mod \pi)= \arg c$$ and therefore $\gamma = e^{2i\arg c}$.
\end{proof}

\begin{rem*}
In the statement of Theorem \ref{t-herm1}, and therefore Theorem \ref{t-herm2}, the constant $c$ cannot be equal to $0$. Indeed, consider $K_{z^n}$ that consists of all polynomials of degree less than $n$. Let $\beta_1,...\beta_{n-2}$ be points on $\T$ such that
$\beta_k^n=\gamma_k$ are different points. Let $$p(z)=a_{n-1}z^{n-1}+...+a_1z$$ be a polynomial with roots at $0,\beta_1,...\beta_{n-2}$.
Then $$\tilde p(z)=\bar a_{1}z^{n-1}+...+\bar a_{n-1}z$$ has roots at the same points. Notice that $p+\tilde p $ is a Hermitian element of $K_{z^n}$
whose zero set $Z=\{0,\beta_1,...\beta_{n-2}\}$ satisfies $\sigma_{\gamma_k}(Z)=1/n>0$ for $k=1,2,...,n-2$.
\end{rem*}

\medskip\noindent Let us mention the following examples that illustrate Theorem \ref{t-herm1}.

\bigskip
\medskip\noindent {\bf Level sets of self-reciprocal polynomials.}
 A polynomial
\[
 p(z)=\sum_{m=0}^{n-1}a_m z^m
\]
 of degree less than $n$  is  a Hermitian element of $K_{z^n}$,
  if and only if
  \begin{align}\label{c-coeff}
  a_0=0\qquad \text{and}\qquad a_m = \overline{a_{n-m}},\qquad m=1,\hdots, n-1.
  \end{align}
  Polynomials that satisfy \eqref{c-coeff} are called \emph{self-reciprocal}.
  Note that the Clark measure $\sigma_\gamma$ of $\theta=z^n$ is concentrated on
  the set of $n$-th roots of $\gamma$.

  \medskip\noindent \textit{Hence, if $c\neq 0$ and $p$ is a self-reciprocal polynomial, then by
 Theorem \ref{t-herm1} all roots of the equation $p=c$ on $\T$ must  be contained
 in a set of $n$-th roots of $\gamma$ for  $\gamma\in\T$ given in the statement of the theorem.
}

\medskip\noindent Naturally, this simple fact can also be proved directly. If $z$ is such that $p(z)=c$ then
$$c = p(z) = z^n \overline{p(z)} = z^n \bar c.$$
 Hence $z^n  = c/\bar c$ where $|c/\bar c| = 1$ and $\arg (c/\bar c) = 2\arg c$. The proof of
 Theorem \ref{t-herm1} can be viewed as a generalization of this argument.

\bigskip
\medskip\noindent {\bf The non-zero level sets of Paley--Wiener functions.}
Recall the definition of Paley--Wiener functions from Subsection \ref{ss-PAL}.

\medskip\noindent The following statement is an analog to Euler's Formula $e^{i\theta} = \cos\theta+i\sin\theta$ for Paley--Wiener functions.

\begin{prop}\label{t-Euler}
Let $f\in \PW_a$.
Then $e^{iaz}f=g_1 +i g_2$ where $g_1,g_2$ are entire functions such that
$g_1,g_2\in e^{iaz}\PW_a$ and each level set
$$\{x\in\R| \ g_i(x)=c\},\ i=1,2;\ c\neq 0,$$
is contained in the arithmetic progression
$$\left\{2\arg c+\frac {2\pi n}a\right\}_{n\in\Z}.$$
\end{prop}

\begin{proof}
Recall that $K_{\vartheta_a}^\R$ was defined as the function space obtained by ``mapping"  the model space $K_\te$, where $$\te_a(z) = \te(z) = e^{-2a\frac{1+z}{1-z}}, \qquad0<a\le 1$$ from $\D$ to $\R$ using the standard conformal map $\psi:\C_+\to\D$, see Section \ref{ss-PAL}.
Then $$e^{iaz}\PW_a\subset K_{\vartheta_a}^\R.$$ Hence we have $e^{iaz}f\in K_\vartheta^\R$.

\medskip\noindent Without loss of generality $a=1$. For the inner function $\te = e^{-\frac{1+z}{1-z}}$ in $\D$, the Clark measure $(\sigma_\gamma)_{\gamma\in \T}$
is concentrated on the sequence $$\psi(\{\arg\gamma+2\pi n\})\subset\T.$$

\medskip\noindent Like in the remark following Theorem \ref{t-herm1}, we can decompose $f$ into $f = g_1+i g_2$ where the translations $\tilde g_1$ and $\tilde g_2$ of the functions $g_1$ and $g_2$ from the upper half plane to the disc are Hermitian in $K_\theta$.
By Theorem \ref{t-herm1}
\[
 \{z\in\T: \tilde g_i = c\} \subset \psi(\{\arg\gamma+2\pi n\}),\ \arg\gamma=2\arg c.
\]

\medskip\noindent Hence the level sets of the functions $g_i$ are contained in  arithmetic progressions given in the statement.
\end{proof}

\providecommand{\bysame}{\leavevmode\hbox to3em{\hrulefill}\thinspace}
\providecommand{\MR}{\relax\ifhmode\unskip\space\fi MR }
\providecommand{\MRhref}[2]{%
  \href{http://www.ams.org/mathscinet-getitem?mr=#1}{#2}
}
\providecommand{\href}[2]{#2}

\end{document}